\documentclass{article}
\pdfoutput=1

\usepackage[left=1.25in,top=1.25in,right=1.25in,bottom=1.25in,head=1.25in]{geometry}
\usepackage{amsfonts,amsmath,amssymb,amsthm}
\usepackage{verbatim,float,url,enumerate}
\usepackage{graphicx,subfigure,epsfig,psfrag}
\usepackage{dsfont,bm,color,appendix}
\usepackage{natbib}
\usepackage{wrapfig}
\usepackage[boxed]{algorithm2e}

\usepackage{hyperref}
\newcommand{\footremember}[2]{%
    \footnote{#2}
    \newcounter{#1}
    \setcounter{#1}{\value{footnote}}%
}

\newtheorem{theorem}{Theorem}

\newtheorem{remark}{Remark}
\newtheorem{corollary}{Corollary} 
\newtheorem{proposition}{Proposition}

\newcommand{\argmin}{\mathop{\mathrm{argmin}}}

\newcommand{\minimize}{\mathop{\mathrm{minimize}}}

\def\E{\mathbb{E}}

\def\diag{\mathrm{diag}}

\def\R{\mathbb{R}}

\usepackage[T1]{fontenc}

\def\S{\mathcal{S}}

\theoremstyle{definition}
\newtheorem{definition}{Definition}

\theoremstyle{prop}
\newtheorem{prop}{Proposition}

\def\simiid{\stackrel{iid}{\sim}}
\newcommand{\norm}[1]{\left\lVert#1\right\rVert}

\title{Log-ratio Lasso: Scalable, Sparse Estimation for Log-ratio Models}
\author{Stephen Bates\footremember{stephen}{Department of Statistics, Stanford University.  Email: stephenbates@stanford.edu.}
 \and Robert Tibshirani \footremember{rob}{Departments of Biomedical Data Science and Statistics, Stanford University. Email: tibs@stanford.edu.}}
\begin{document}
\maketitle

\begin{abstract}
Positive-valued signal data is common in many biological and medical applications, where the data are often generated from imaging techniques such as mass spectrometry. In such a setting, the relative intensities of the raw features are often the scientifically meaningful quantities, so it is of interest to identify relevant features that take the form of log-ratios of the raw inputs. When including the log-ratios of all pairs of predictors, the dimensionality of this predictor space becomes large, so computationally efficient statistical procedures are required. We introduce an embedding of the log-ratio parameter space into a space of much lower dimension and develop efficient penalized fitting procedure using this more tractable representation. This procedure serves as the foundation for a two-step fitting procedure that combines a convex filtering step with a second non-convex pruning step to yield highly sparse solutions. On a cancer proteomics data set we find that these methods fit highly sparse models with log-ratio features of known biological relevance while greatly improving upon the predictive accuracy of less interpretable methods.
\end{abstract}

\section{Introduction}

In biological and medical sciences, imaging techniques are a common tool for investigating the molecular composition of a sample of interest. These methods give a intensity values across a range of frequency values, and a data processing step translates the raw spectrum into intensity values for many molecules of interest. With such data, meaningful standardization of the resulting features is difficulty, so it is desirable to treat this as compositional data.  In a medical imaging setting, a common task of interest is to predict disease status from a tissue sample.  Recent results by \cite{banerjee-et-al} show that such data  can be used to detect the presence of prostate cancer from tissues samples with high accuracy.  

In both biological and medical settings, it is desirable to have a model that combines high predictive accuracy with interpretability and parsimony.  Parsimonious models can lead to (i) novel scientific insights (ii) increased ease of model checking and confidence in the model.  This second point is especially important in a medical context where a method may be deployed for use in a hospital setting. Deployment requires very high confidence in the effectiveness and robustness of a predictive learning model.  Model interpretability allows domain experts to add another layer of validation to complement traditional statistical validation using a training set and a test set.  As a result, such a method is much more likely to be trusted by safety advisory boards that mush approve deployment of new treatment methods.

Variable selection in regression models is a well-studied topic in the statistical literature.  Classically, best subset selection, forwards stepwise selection, and backward stepwise selection are commonly used methods for variable selection [see e.g. \cite{esl}].  To determine the size of the final model selected, a practitioner may use Mallow's Cp/AIC [\cite{aic}, \cite{cp}], BIC [\cite{bic}], or cross-validation prediction error.  A more modern approach that has enjoyed widespread success is the lasso estimator introduced by \cite{lasso}.  The lasso uses $L_1$-penalization to induce sparse models, an approach that subsequently has been fruitfully applied in a variety of other variable selection contexts.  Similarly, methods for compositional data are well developed in the statistical literature [\cite{compositional-data}]. Compositional data is usually transformed to real-valued data using a logarithmic transform after either (i) transforming all features to their relative intensity to an arbitrarily chosen baseline feature [\cite{log-contrasts}] or (ii) standardizing each observation by the geometric mean of the features [\cite{compositional-pca}].  \cite{constrained-lasso} utilize this second standardization scheme together with $L_1$-penalization to do variable selection with compositional data.  

\subsection{Our Contribution}
In this work, we introduce a new notion of sparsity for compositional data based on the log-ratio model. For sparse estimation, we use a $L_1$-penalized regression as introduced in \cite{lasso}.  We prove the equivalence of our method with that of a modified lasso estimator in a lower-dimensional space. We then introduce a second pruning step to fit even sparser models after giving a mathematical argument for why this additional step is required.  Our mathematical characterization of the log-ratio feature space also gives rise to an approximate forward stepwise algorithm for this same setting.  To the best of the authors' knowledge, there are no existing efficient fitting algorithms for these models.  We further use this characterization to develop a novel post-selective goodness-of-fit test.  In a simulation study, we find that when the true underlying model has log-ratio structure, our estimator greatly improves upon generic approaches.  When applied to a medical proteomics data set, the method improves the predictive performance relative to baseline methods and discovers features of biological importance.

\subsection{Outline}
This paper proceeds as follows.  In section \ref{sec:lr-models} we introduce log-ratio models and give a brief exposition of their properties.  In section \ref{sec:lrl-estimator}, we propose an estimator for this family of models, and formally establish the relationship between our estimator and a standard lasso estimator in a larger space.  We then develop an additional pruning stage which leads to much sparser model fits. In section \ref{sec:afs} we propose an approximate forward stepwise algorithm as an additional scalable estimator for this family of models.  In section $\ref{sec:testing}$, we explain how to test the appropriateness of a log-ratio model using classical testing and develop a novel post-selective test of model fit.  In section \ref{sec:sims} we examine these estimators in simulation experiments.  In section \ref{sec:real-data}, we apply these estimators to a proteomics data set to diagnose cancer from metabolite intensity.  Lastly, we offer a few concluding remarks in section \ref{sec:discussion}.

\section{Log-ratio Models}\label{sec:lr-models}
We begin by introducing log-ratio models and briefly explain the appeal of such models. Let $Y = (y_1,...,y_n)$ be the response variable and let $X$ be the $n \times p$ matrix of positive features, which does not include an intercept term.  In the remainder of this work, we will be interested in models of the form
\begin{equation}
\label{eq:logratio}
y_i =  \mu + \sum_{1 \le j < k \le p}\theta_{j,k}\log(x_{i,j} / x_{i,k}) + \epsilon_i
\end{equation}
for $i = 1,...,n$.  The $\epsilon_i$ are mean zero noise terms with common variance.  In particular, we are interested in fitting such models under the assumption that most of the entries $\theta_{j,k}$ for $1 \leq j < k \leq p$ are zero.  If we let $Z$ be the $n \times \binom{p}{2}$ matrix with columns $\log(x_{j} / x_{k})$ for $1 \leq j < k \leq p$ then we see that this is a linear model in $Z$.
\begin{equation*}
Y = Z \theta + \epsilon.
\end{equation*}
Such models are attractive for several reasons:

\begin{enumerate}[(a)]
\item {\bf Predictive accuracy.}
In settings such medical imaging, we may expect that the {\em relative} intensity a small number of raw features is what carries information about the response variable. Furthermore, in such imaging settings, the distribution of the raw features is often very skewed with a few very large values and many small values.  In this case, a logarithmic transform of the data can make the feature matrix less skewed, improving the predictive performance of the model.

\item{\bf Interpretability and scientific relevance.}
For reasons related to those above, in some fields (e.g. biochemistry) it is standard to work with log-ratio quantities in the course of a scientific investigation.  In such a case, a statistical model that consists of a few log-ratios has much higher interpretability than other approaches.  Researchers with domain knowledge can more easily understand both the predictors that are selected and the effect size associated with these predictors. This enhances scientific insight, and the increased transparency of the model leads to greater robustness and confidence in the modeling process.

\item{\bf Robustness to standardization choices.}
Log-ratio models are a model for compositional data: they are invariant to the multiplicative scaling of each observation.  The prediction for observation $i$ depends only on terms of the form $x_{i,j} / x_{i,k}$, which remain unchanged when the whole feature vector $(x_i)_{1=1,...,p}$ is scaled by some constant.  This property is relevant with positive signal data, since the processing of the data from raw functional form to intensities of various molecules requires arbitrary choice of scale {\em for each observation}.  We emphasize that standardizing the {\em features} in later stages does not address this problem.  Although we do not often explicitly model the distribution of the features $X$, there is of course some noise in these values.  The choice of scaling is implicitly multiplying each observation by some random constant. This is introducing additional noise into the model, which we can roughly interpret as increasing the variance $\sigma^2$ of the noise terms.  Methods that are invariant to this scaling avoid the resulting increase in variance.  Parameterizing the problem in terms of all log-ratios is appealing because it both maintains symmetry and allows for easy interpretation.
\end{enumerate}

\section{The Log-ratio Lasso Estimator}\label{sec:lrl-estimator}
Motivated by the linear model formulation in (\ref{eq:logratio}), we next propose an estimator for the sparse log-ratio model setting and examine its properties.  
  
\begin{definition}[Log-ratio lasso]
The {\it log-ratio lasso} fit is defined to be:
\begin{equation}\label{eq:lr-lasso}
\theta \in \argmin_{\theta} \frac{1}{2}\sum_{i = 1}^n [y_i - \mu - \sum_{1 \le j < k \le p}\theta_{j,k}\log(\frac{x_{i,j}}{x_{i,k}})]^2 + \lambda \norm{\theta}_1.
\end{equation}
\end{definition}
Here $\theta_{i,j}$ is defined only for $i < j$, so this can be represented as a vector in $\R^{\binom{p}{2}}.$  $\lambda$ is a tuning parameter.  This is the usual lasso estimator with the expanded feature matrix $Z$.  It is a well-known that the solutions $\theta$ will be sparse for appropriate values of $\lambda$.  This estimator is statistically quite natural, but solving this optimization problem is challenging because it is an optimization in $\binom{p}{2}$ variables and requires explicit construction of the $n \times \binom{p}{2}$ matrix of predictors $Z$. For even moderate values of $p$, this would require a large amount of memory and computation time.  For large values of $p$, this becomes intractable.  

\subsection{Low-dimensional Characterization of the Log-ratio Lasso}
The log-ratio lasso is an appealing variable selection technique for our setting, but in its original form it is not practical for large $p$ because solving such an optimization problem in $\binom{p}{2}$ variables will have prohibitively large running time and memory requirements.  A major component of our contribution is to show that this estimator can be recovered from the solution of a much simpler optimization problem in only $p$ variables, that does not require the construction of the expanded feature matrix $Z$. We proceed by defining a simpler optimization problem and connecting it to the log-ratio lasso estimator.

\begin{definition}[Linearly constrained lasso]
The {\it linearly constrained lasso} fit is defined to be:
\begin{equation} \label{eq:constr-lasso}
\beta = \argmin_{\beta: \sum_{j=1}^p \beta_j = 0} \frac{1}{2}\sum_{i = 1}^n [y_i - \mu - \sum_{j = 1}^p \beta_j \log(x_{i,j})]^2 + \gamma \norm{\beta}_1.
\end{equation}
\end{definition}
Notice this is a standard lasso on the logarithmically transformed variables, with that additional constraint that the sum of the coefficients is zero. This is a convex optimization problem with $p$ variables and 1 linear constraint, and can be solved efficiently. We provided details about how to efficiently solve this optimization problem using any standard lasso solver in the Appendix.  This estimator was first studied in \cite{constrained-lasso}, also in the context of compositional data.  In that work, however, the authors were not interested in the log-ratio models (\ref{eq:logratio}).  This work develops a different notion of sparsity.  We now turn to the relationship between this estimator and the log-ratio lasso estimator.

\begin{theorem}
The log-ratio lasso problem and the linearly constrained lasso problem are equivalent for $\lambda = 2\gamma$. 
\label{thm:main}
\end{theorem}

This theorem is important because implies that we can fit the log-ratio lasso efficiently, even for large values of $p$. 

\subsubsection*{Definition of equivalence}
We must explain in detail what we mean by ``equivalent'' in the statement of the previous theorem.  Let $W = \log(X)$ be the matrix obtained by taking the element-wise logarithm of $X$. Let $b: \R^{\binom{p}{2}} \to \R^p$ be the linear map that takes a $\theta$ from a log-ratio lasso feature space to the corresponding $\beta$ in the standard feature space:
\begin{equation*}
b(\theta)_k = \sum_{j = 1}^{k-1} - \theta_{j,k} + \sum_{j = k + 1}^p \theta_{k,j}.
\end{equation*}
This definition implies $Z \theta = W[b(\theta)]$. The log-ratio lasso fit is not unique, so ``equivalence'' in the preceding theorem means that the same minimum value of the objective function is achieved for both functions, and for any solution $\theta$ of the log-ratio lasso optimization, $\beta = b(\theta)$ is a solution of the linearly constrained lasso optimization. By the following corollary, the solution to the linearly constrained lasso optimization will typically be unique when $n > p$.

\begin{corollary}
The minimizer of the constrained lasso is unique if the matrix $W$ has full rank. 
\end{corollary}
\begin{proof}
From the KKT conditions of the standard lasso optimization, one can show that all solutions $\theta$ of the log-ratio lasso have equivalent fitted values $Z\theta$ [see e.g. \cite{tibs-lasso-uniqueness}].  Since $W$ has full rank, there is a unique $\beta$ such that $Z\theta = W\beta$ and the result follows.
\end{proof}

\subsubsection*{Extensions beyond the linear model}
We note in passing that this theorem also applies to $L_1$-penalized GLMs[\cite{glm}] and the $L_1$-penalized Cox proportional hazards model[\cite{cox-model}, \cite{cox-lasso}], which can be seen from the argument in the Appendix.  It would be of great interest to further pursue the case of logistic regression, which is the natural model when dealing with binary outcome data such as disease status.  Similarly, the Cox model is of great interest because survival time outcomes are common in medical contexts.  For the remainder of this work, however, we limit ourselves to $L_1$-penalized linear regression.

\subsubsection*{Mathematical intuition}
We now give some intuition for the preceding theorem.  A formal proof is provided in the Appendix.  Notice that any model $Y = Z\theta + \epsilon$ corresponds to a model $Y = W\beta + \epsilon$ with the mapping $\beta = b(\theta)$ from the preceding paragraph.  Furthermore, we have that $\sum_{k=1}^p \beta_k = \sum_{k=1}^p[\sum_{j=1}^{k-1} -\theta_{j,k} + \sum_{j = k+1}^p \theta_{k,j}] = 0$, since each term $\theta_{j,k}$ for $j < k$ appears once with positive sign and once with negative sign.  Thus when searching for a model in the form $W\beta$, we can restrict out attention to vectors $\beta$ such that $\sum_{k=1}^p\beta_k = 0$.

Take $\beta = b(\theta)$.  We now turn to the connection between $\norm{\theta}_1$ and $\norm{\beta}_1$.  It is not true that $2 \norm{\theta}_1 = \norm{\beta}_1$ in general, but it does hold in some cases.  Loosely speaking, this relationship holds whenever $\theta$ does not lead to redundant linear combinations of columns.
In the Appendix we show that among the vectors $\theta^\prime$ such that $Z\theta = Z\theta^\prime$, the minimum $L_1$-norm solution satisfies $2 \norm{\theta^\prime}_1 = \norm{\beta}_1$.  For other choices of  $\theta^\prime$, the $L_1$ norm can be reduced without changing the model fit, and those values of $\theta^\prime$ will not be selected by the log-ratio lasso procedure.  We then show that for relevant $\beta$, there exists a corresponding $\theta$ with this property, and the theorem follows.

\subsection{Relationship with the lasso}
From this formulation of the problem, we can see that the log-ratio lasso is finding a model of the form 
\begin{equation}
\label{eq:loglinearmodel}
\hat{y}_i = \mu + \sum_{j=1}^p \beta_j \log{x_{i,j}} \ (i = 1,...,n).
\end{equation}
The log-ratio lasso differs from the standard lasso in that it is searching for models that can be represented as a sparse collection of ratios, rather than only a sparse coefficient vector $\beta$.  The lasso on the transformed set of features $(\log{x_{i,j}})_{j=1,...,p}$ and the log-ratio lasso are both fitting models within this family, but they differ in the way that they search among the candidate models.  We will see in simulation experiments that when the true model consists of a few log-ratios, the specialized log-ratio lasso has better performance.  This is analogous to how lasso and ridge regression sometimes have very different performance, despite the fact that both methods are fitting linear models from the same model family.

\subsection{Non-uniqueness of solution}
There are linear dependencies among the features in the expanded feature matrix $Z$ because there are $\binom{p}{2}$ features takes of form $\log(x_i) - \log(x_j)$, which together have a linear span of only dimension $p - 1$.  Since $Z$ is not full rank, the OLS solution regressing $Y$ onto $Z$ is not well-defined.  A standard way to fix such non-identifiability is to add a penalty term.  Interestingly, in the model $Y = Z\theta + \epsilon$, the addition of $L_1$-penalization is not sufficient to make the model fit unique.

\subsubsection*{Explicit example}
To demonstrate the non-uniqueness in the solutions $\theta$, suppose we have fit the following model:
\begin{equation*}
\hat{y} = 2 \log(x_1) + 1 \log(x_2) - 2 \log(x_3) - 1\log(x_4).
\end{equation*}
We can represent this as $y = Z\theta$ for several different values of $\theta$ with equivalent $L_1$ norm:

\begin{enumerate}[(a)]
\item $\hat{y} = 2 [\log(x_1) - \log(x_3)] + 1 [\log(x_2) - \log(x_4)]$. Here $\norm{\theta}_1 = 2$.

\item $\hat{y} = 1.5 [\log(x_1) - \log(x_3)] + .5 [\log(x_1) - \log(x_4)] + 
.5 [\log(x_2) - \log(x_3)] + .5 [\log(x_2) - \log(x_4)]$.  Here $\norm{\theta}_1 = 2$.

\item $\hat{y} = 1.7 [\log(x_1) - \log(x_3)] + .3 [\log(x_1) - \log(x_4)] + 
.3 [\log(x_2) - \log(x_3)] + .7 [\log(x_2) - \log(x_4)]$. Here $\norm{\theta}_1 = 2$.

\end{enumerate}

This ambiguity in $\theta$ occurs whenever there are at least 2 disjoint log-ratios that are nonzero. The solution $\theta$ may no be unique unique, but we noted earlier that $b(\theta)$ is unique unless the matrix $W$ is not of full rank. Although there may exist many equivalent solutions $\theta$ with the same $L_1$-norm, we see that some of these solutions may be sparser than others. Since even $L_1$-penalization is not enough to enforce the desired sparsity in $\theta$, a more demanding selection criterion is needed. We now turn to a method for finding a highly sparse solution $\theta$.

\subsection{Two-stage procedure} \label{subsec:two-stage}
Our goal is to find a succinct model of the form (\ref{eq:logratio}) that explains the data well. In the preceding section, we saw that there is some ambiguity resulting from the solution to the log-ratio lasso optimization problem.  We now turn to a method for finding a highly sparse parameter $\theta$.  We will use the log-ratio lasso solution as a baseline solution, followed by a secondary sparse regression procedure: \\

\begin{algorithm}[H] \label{alg:two-stage}
{\bf Algorithm \ref{alg:two-stage}:} Two-stage Log-ratio Lasso
\begin{enumerate}
\item Fit the log-ratio lasso with tuning parameter $\lambda$.
\item Using variables in the support from step one, enumerate all log ratios into a matrix $\tilde{Z}$.
\item Use the sparse regression procedure to fit $y$ onto $\tilde{Z}$, with tuning paramemter $\gamma$.
\end{enumerate}
\end{algorithm} \ \\

 Examples of such procedures include forward stepwise, backward stepwise, or best subset selection.  There are several other sub-$L_1$ methods explored in the literature such as the $MC+$[\cite{mcp}] and SCAD[\cite{scad}] penalties, and $MC+$ penalized linear regression is readily available in the {\tt sparsenet} [\cite{sparsenet}] R package. In algorithm \ref{alg:two-stage}, the parameters $\lambda$ and $\gamma$ should be chosen by joint cross-validation over a grid of values.  Using warm-starts for the solutions of the two optimization problems in step 1 and step 3 means that this takes much less time than fitting the model from scratch at each value of the parameter.  When using forward-stepwise regression in step 3, the model fitting is particularly fast. In this paper, we concentrate on forward-stepwise regression, because it works in the GLM case and scales to large data sets.

We can view this as a class of regression procedures that produce models that are even sparser than the log-ratio lasso. Due to the size of predictor space ($\binom{p}{2}$ variables) we use lasso as a screening step, followed by a `sparser' regression procedure.  In some cases fitting something as simple as a forward stepwise regression of $y$ onto the original $Z$ would be computationally expensive, so this screening step is necessary.  This algorithm provides both a statistical and computational improvement over forward stepwise selection for large $p$.  

As an alternative, one can fit the sparse regression procedure in step 3 using the predicted values $\hat{y}$ from step 1 instead of the observed values $y$.  This will result in a final model that is more similar to the single-stage log-ratio lasso fit, but the terms will be paired into ratios for easier interpretation.  We will refer to this variant as {\em conservative two-stage log-ratio lasso}, since by maintaining the shrinkage from the first stage it will usually have less variance (at the price of more bias) than the standard two-stage procedure.  We empirically study the performance of both estimators in simulation experiments in section \ref{sec:sims}. 
 
\subsection{Including unpaired logarithm terms}
In some circumstances, it may be desirable to search for a model that includes both log-ratios and unpaired log terms:
\begin{equation*}
y_i =  \sum_{j=1}^p \beta_j \log(x_j) + \sum_{j < k}^p\theta_{j,k}\log(x_{i,j} / x_{i,k}).
\end{equation*}
The span of this model is again contained in the span of model (\ref{eq:loglinearmodel}), but in this case the analyst wishes to favor the selection log-ratio terms wherever possible. Fitting this model requires only a simple modification of the log-ratio lasso.  We augment the feature matrix with the vector of ones: $x_{p+1} = 1$.  Since $\log(x_i / x_{p+1}) = \log(x_i)$, the log-ratio lasso with the augmented feature matrix can select unpaired terms whenever it is beneficial to do so.  This model fit gives a compromise between the standard lasso on the logarithmically transformed features and the log-ratio lasso.

\section{Approximate Forward Stepwise Selection} \label{sec:afs}

We have introduced a procedure to efficiently fit large log-ratio models based on the lasso.  We now introduce an alternative greedy fitting procedure that also takes advantage of the unique log-ratio structure of our setting.  In analogy with the log-ratio lasso, one might try simply running forward stepwise selection on the expanded feature set consisting of all log ratios.  As in the previous section, running forward stepwise selection on the expanded feature matrix $Z$ of size $n \times \binom{p}{2}$ quickly becomes infeasible for large $p$. We now show how the log-ratio structure gives rise to a simple modification of the forward stepwise selection procedure to allow for efficient log-ratio model fitting.  We first describe the algorithm and then examine the relationship between this procedure and standard forward stepwise selection applied to the expanded feature matrix $Z$.\\

\begin{algorithm}[H] \label{alg:afs}
{\bf Algorithm \ref{alg:afs}:} Approximate Forward Stepwise Selection \\ \ \\
Standardize the features to have mean 0 and variance 1.  Begin with $\S = \emptyset$ and $r = y - \bar{y}$.

\For{j =  1 to k}{
	\begin{enumerate}
	\item Fit a single-variable linear regression on features $1$ to $p$ using residual $r$ as the \\ response.
	\item Select the feature $\log(x_i / x_j)$ where $i$ corresponds to the feature with largest positive coefficient and $j$ corresponds to the feature with largest negative coefficient.  Add this  \\ ratio to the set of selected ratios $\S$.
	\item Regress out the log-ratio features in $\S$ together with an intercept from $y$ to obtain a \\ new residual vector $r$.
	\end{enumerate}
}
\end{algorithm} \ \\

We now explain in what sense this procedure approximates forwards stepwise selection. For linear regression on all log-ratio terms, at each step forward stepwise selection will choose $i$ and $j$ to maximize
\begin{equation} \label{eq:fs}
\frac{r^\top (\log(x_i) - \log(x_j))}{\sqrt{\norm{\log(x_i)}^2 + \norm{\log(x_j)}^2 + 2 \log(x_i)^\top \log(x_j)}}
\end{equation}
whereas approximate forward stepwise will choose $i$ and $j$ to maximize
\begin{equation*}
\frac{r^\top (\log(x_i) - \log(x_j))}{\sqrt{\norm{\log(x_i)}^2 + \norm{\log(x_j)}^2}}
\end{equation*}
where $r$ denotes the residual at that step.  Approximate forward stepwise procedure selects a log-ratio with high explanatory power at each step, ignoring correlation among the features to make the computation feasible.  We note that when the vectors $\log(x_i)$ are mutually orthogonal for $i = 1,...,p$, the proposed approximate stepwise procedures coincides with the exact forward stepwise solution on the expanded feature set $Z$. The smaller the correlations are among the $\log(x_i)$, the more likely that this procedure will make the same selection as forward stepwise selection. 

Equation $\ref{eq:fs}$ shows that we can fit forward stepwise faster than usual in this setting, since we don't need to project out the selected feature from all approximately $\binom{p}{2}$ features at each step, as is usually done [see e.g. \cite{esl}]. The approximate forward stepwise procedure is even faster than forward stepwise regression on the expanded feature set $Z$: at each step both procedures require the computation of $p$ inner products, but standard forward stepwise also requires the computation of the $\binom{p}{2}$ inner products among the features.  When the total number of steps $k$ is small relative to $p$, the complexity for approximate forward stepwise selection will be $O(npk)$ and for forward stepwise selection will be $O(np^2)$.  We examine the runtime and and statistical performance of this procedure in simulations in section \ref{sec:sims} and on a real data set section \ref{sec:real-data}.

\section{Model Verification via Hypothesis Testing} \label{sec:testing}
In section \ref{sec:lr-models} we motivated the use of log-ratio models for positive signal data. We proceeded to develop estimation methods for these models in sections \ref{sec:lrl-estimator} and \ref{sec:afs}. For a practitioner applying such methods, it is essential to have model diagnostics for evaluating the resulting model fits.  In section \ref{sec:lrl-estimator} and the associated proofs in the Appendix, it is shown that a linear model in the logarithmically transformed variables (\ref{eq:loglinearmodel}) corresponds to a log ratio model (\ref{eq:logratio}) if and only if 
\begin{equation} \label{eq:sum-coef-zero}
\sum_{i=j}^p \beta_j = 0.
\end{equation}
As a goodness-of-fit test of the log-ratio model, a practitioner can conduct a formal hypothesis test of (\ref{eq:sum-coef-zero}).  We present two tests, a classical test and a more modern selective inference test.

\subsection{Classical Testing}
If $p < n$ and the model matrix of the logarithmically transformed variables is of full rank, then we can use OLS to fit a linear model of the form (\ref{eq:loglinearmodel}).  Under the assumption that the errors $\epsilon$ are i.i.d. normally distributed, we can test the linear hypothesis given by (\ref{eq:sum-coef-zero}) using the F-test. Testing linear hypothesis of coefficients from OLS regression is common in the statistical literature, and can be carried out in R using e.g. the {\tt linear.hypothesis} function in the {\tt car} package [\cite{car-package}].

\subsection{Selective Inference} \label{subsec:selective-inf}
When doing model selection, classical hypothesis testing no longer applies since the object of inference is based on the (random) data.  The selective inference formalism has been developed in recent years to address the problem of inference in the presence of model selection [\cite{post-selection}, \cite{lee-et-al}]. In analogy with classical goodness-of-fit testing in nested models, we will use results from selective inference theory to essentially compare the one-step log-ratio lasso estimator to the standard lasso estimator on the feature $(\log{x_j})_{j=1,..,p}$.  Our proposed test considers the sparse model selected by the lasso rather than the less relevant full model, and applies even when $p > n$.  The proposed test does {\em not} assume correctness of the linear model, only the Gaussianity of the error terms. 

Let $M$ a subset of features together with signs for each, and let $X_M$ be the the feature matrix containing only the selected features.  Let $\beta_j^{(M)} := \E[(X_M^\top X_M)^{-1} X_M^\top y]$ be the partial regression coefficients in a model consisting only of features in $M$.  We will test the hypothesis that 
\begin{equation} \label{eq:selective-hypothesis}
\sum_{j\in M} \beta^{(M)}_j = 0
\end{equation}
where $M$ is the support set identified by the lasso. With this hypothesis we are asking the question: given the selection of the lasso, is there any log-ratio model in the selected variables consistent with the data?  In situations where a researchers is considering a log-ratio model, this is a sensible question to ask. If there is sufficient resolution in the data to reject this hypothesis, this indicates to the researcher that the log-ratio model is not appropriate and implies that the data is not purely compositional data.

\begin{proposition}[Post-selective test of the log-ratio model]\label{prop:p-values}
Let $F_{\mu, \sigma^2}^{[a,b]}$ be the CDF of a $N(\mu, \sigma^2)$ random variable truncated to the set $[a,b]$.  Then there exist quantities $a(y)$ and $b(y)$ independent of $\eta_M^\top y$ such that
\begin{equation} \label{eq:p-value}
F_{1^\top \beta^{(M)}, \norm{\eta_M}^2}^{\cup_s[a(y), b(y)]}(\eta_M^\top y) | \{\hat{M} = M\} \sim \text{Unif}(0,1).
\end{equation}
\end{proposition}

This proposition gives a conditional p-value for the post-selective hypothesis given in (\ref{eq:selective-hypothesis}):
\begin{equation*}
p_0 = F_{0, \norm{\eta_M}^2}^{\cup_s[a(y), b(y)]}(\eta_M^\top y).
\end{equation*}
The pivotal quantity from (\ref{eq:p-value}) is monotone decreasing in the Gaussian mean parameter $1^\top \beta^{(M)}$, so this is a one-sided p-value.  For our context, we usually prefer the two-sided p-value given by $p_0^\prime = 1 - 2|p_0 - \frac{1}{2}|$. Technical details are provided in the Appendix.

\subsubsection*{Numerical example}
\begin{figure} 
\begin{center}
\includegraphics[scale = 1]{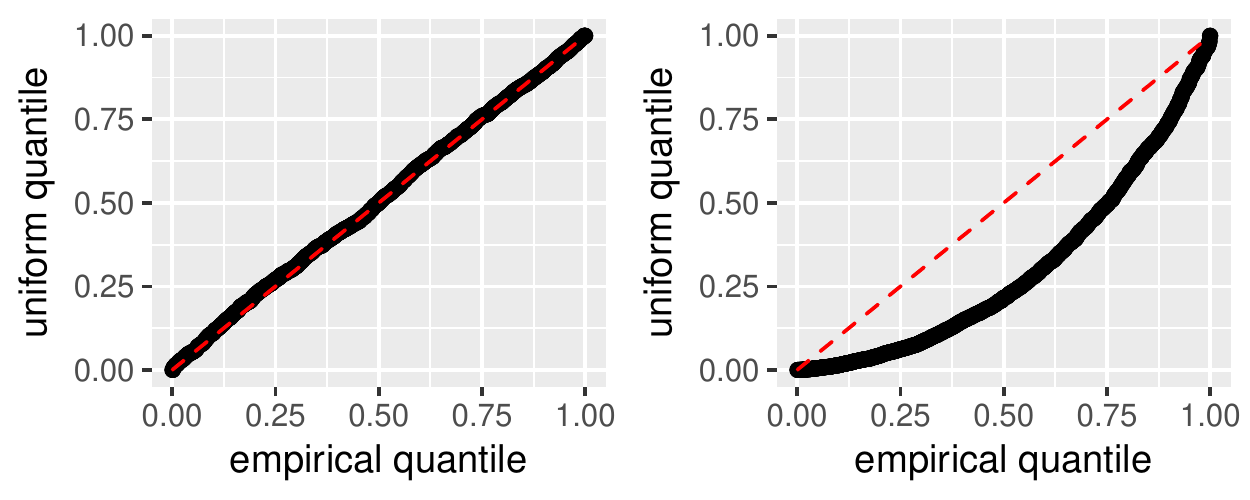}
\end{center}
\caption{{\em Results from the numerical example of subsection \ref{subsec:selective-inf}. The left panel corresponds to the log-ratio model with one signal, in which case the null hypothesis holds whenever the first two features are both selected.  The right panel corresponds to a single unpaired signal, in which case the null hypothesis does not hold.}}
\label{fig:p-values}
\end{figure}

We now examine the post-selective test in a simple setting.  We take $n=100, p=30$ and $X_{i,j} \simiid N(0,1)$.  We use the feature matrix $Z_{i,j} = \log(|X_{i,j}|)$ and generate $y = Z \beta + \epsilon$ with $\epsilon \sim N(0,I_n)$.  We examine the results with $\beta = (2,-2,0,...,0)^\top$ which corresponds to a model with a single large log-ratio signal $\log(\frac{X_1}{X_2})$ and $\beta = (2, 0, ... ,0)^\top$ which is not equivalent to any log-ratio model.  The signals are large enough that the true support is always contained in the selected support. One-sided p-values are computed using proposition \ref{prop:p-values}. We present the result of 2000 repetitions in figure \ref{fig:p-values}.  We observe close agreement with the theory; when the log-ratio model holds the p-values appear to be uniform and when the log-ratio model does not hold the p-values are noticeably sub-uniform.

\section{Simulations}\label{sec:sims}
We now examine the performance of our proposed methods for fitting log-ratio models with simulation experiments.  We examine the following seven methods:

\begin{enumerate}
	\item {\em Approximate forward stepwise} (approx-fs): the approximate forward stepwise procedure described in algorithm $\ref{alg:afs}$. 
	\item {\em Forward stepwise selection} (fs): forward stepwise selection applied on the logarithmically transformed features.
	\item {\em Ridge regression} (ridge): ridge regression applied to the logarithmically transformed features.
	\item {\em Single-stage log-ratio lasso} (single-stage): the method described in (\ref{eq:lr-lasso}).
	\item {\em Two-stage log-ratio lasso} (two-stage): algorithm \ref{alg:two-stage} using forward stepwise selection for the pruning stage.
	\item {\em Two-stage log-ratio lasso} (two-stage-conservative):  the variation of algorithm \ref{alg:two-stage} described at the end of subsection \ref{subsec:two-stage}, again using forward stepwise selection for the pruning stage. 
	\item {\em Lasso} (vanilla-lasso): the usual lasso estimator on the logarithmically transformed raw features.
\end{enumerate}
All tuning parameters are chosen by cross-validation.  Methods 2,3, and 7 are fitting linear models in the logarithmically transformed feature space of the form (\ref{eq:loglinearmodel}) whereas methods 1,4,5, and 6 are fitting log-ratio models of the form (\ref{eq:logratio}).

\subsection{Experiment 1: Two Log-ratio Signals} 

We first examine the performance of our estimator when the data is generated from a log-ratio model.  We consider the following model, consisting of two log-ratio terms of different amplitudes:
\begin{equation*}
y_i = 2 s \log(\frac{x_{i,1}}{x_{i,2}}) + s \log(\frac{x_{i,3}}{x_{i,4}}) + \epsilon_i \text{ for } i = 1,...,n.
\end{equation*}
We take $\epsilon_i \simiid N(0,1)$.  In the following simulations we use $n = 100$, $p = 30$, $X _{i,j} \simiid |N(0, 1)|$.  The signal strength $s$ is taken across a grid of values from 0 to 3.  We present the result in figure \ref{fig:two-signal}.

\begin{figure}
\begin{center}
\includegraphics[scale = .95]{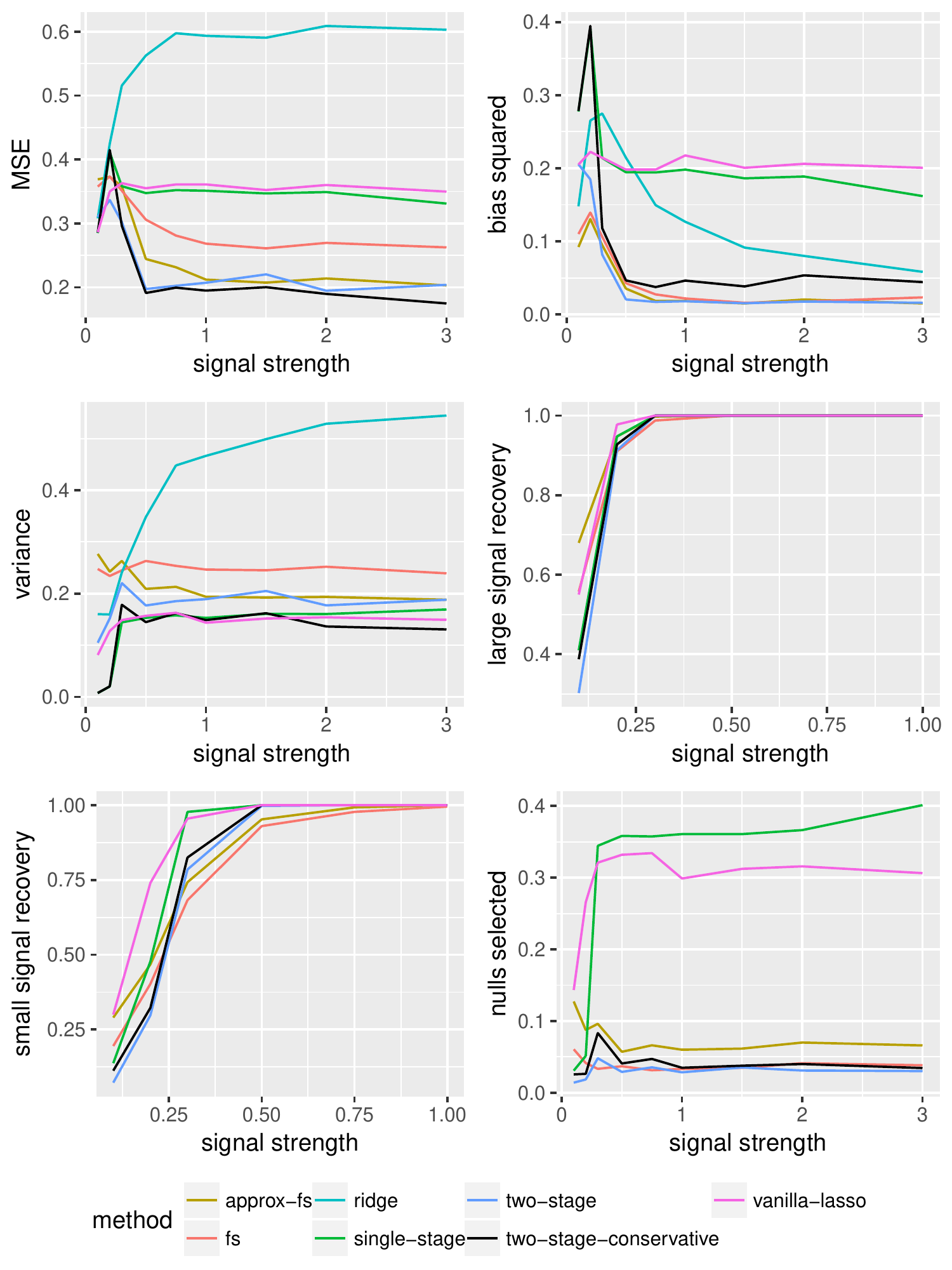}
\caption{{\em Results of experiment 1: MSE and support recovery of log-ratio lasso in the sparse log-ratio model. The ``large signal recovery'' and ``small signal recovery'' graphs report the proportion of times that the true large signal and true small signal are selected, respectively.  The ``nulls selected'' graph shows the average fraction of null variables that are selected.}}
\label{fig:two-signal}
\end{center}
\end{figure}

\subsubsection*{MSE, bias, and variance}
We find that there is a large regime of signal strengths where the two-step procedure preforms significantly outperforms the original lasso and ridge regression.  For coefficients from .5 to about 3, there is a MSE reduction of about 40\% relative to the lasso. The two step procedure has very low bias in the sparse setting, because it sets many coefficients to zero.  It has slightly more variance than lasso, due to its discontinuous nature.  We note that the conservative two-stage procedure has lower variance and higher bias than the two-stage procedure, as expected.  We also note that the lasso and the single-stage log-ratio lasso are quite close, with the single-stage log-ratio lasso performing slightly better. This is expected; the single-stage log-ratio lasso has one extra piece of true information built into the procedure, the fact that the sum of the coefficients must be zero.  Approximate forward stepwise selection has competitive MSE to the two-stage procedures.  The approximate forward stepwise has superior performance to standard forward stepwise in this case, because approximate forward stepwise is picking out log-ratios, whereas forward stepwise is choosing single predictors.  Overall, the two-stage procedures have the best performance in terms of MSE.

\subsubsection*{Support recovery}
We next consider the support recovery properties of these procedures. Ridge regression is fitting a dense model, so it is omitted from the following discussion. The lasso and single-stage procedure recovers the signals slightly more often than the two-step procedure. This is expected, because the two-step procedure is a pruning of the single stage procedure. The two-step procedure selects very few null variables. This explains why this procedure has much better MSE and is an appealing aspect of this procedure in scientific contexts. The approximate forward stepwise procedure selects slightly more nulls than the two-stage procedures, and recovers the true signals slightly less frequently. Forward stepwise selection selects  roughly the same number of non-nulls as the two-stage procedures, but selects the true signals slightly less often.  Overall, the two-stage procedures have the best support recovery properties; these procedures recover the true signals very frequently and rarely select null variables.

\subsection{Experiment 2: Robustness to Model Misspecification}
We now generate data from a model that does not consist only of log-ratio terms.  The data generating process is now:
\begin{equation*}
y = 2 s \log(\frac{x_1}{x_2}) + s \log(\frac{x_3}{x_4}) + .3 \log(x_5) + \epsilon.
\end{equation*}
Notice the inclusion of an unpaired raw term $.3 \log(x_5)$, so in this case we say the log-ratio model is misspecified.  The amplitude of additional term is chosen to be large enough that it can be detected with high probability by the standard cross-validated lasso.  We present the results in figure \ref{fig:mispec}.

Even in the misspecified setting, we see that the two-step procedure has low MSE, again significantly outperforming both lasso and ridge regression. Forward stepwise selection has slightly better MSE that the two-stage procedures, and the approximate forward stepwise procedure has slightly worse MSE. We again see that the lasso and the single-stage procedure have similar performance.  In this case, forward stepwise selection selects the fewest null variables, followed closely by the two-stage procedures and approximate forward stepwise selection. The two stage procedures recover the true signals slightly more often than the forward stepwise and approximate forwards stepwise procedures.  This simulation study gives us some confirmation that even in the presence of moderate model misspecification, the log-ratio lasso procedure will retain good performance.

\begin{figure}
\begin{center}
\includegraphics[scale = .95]{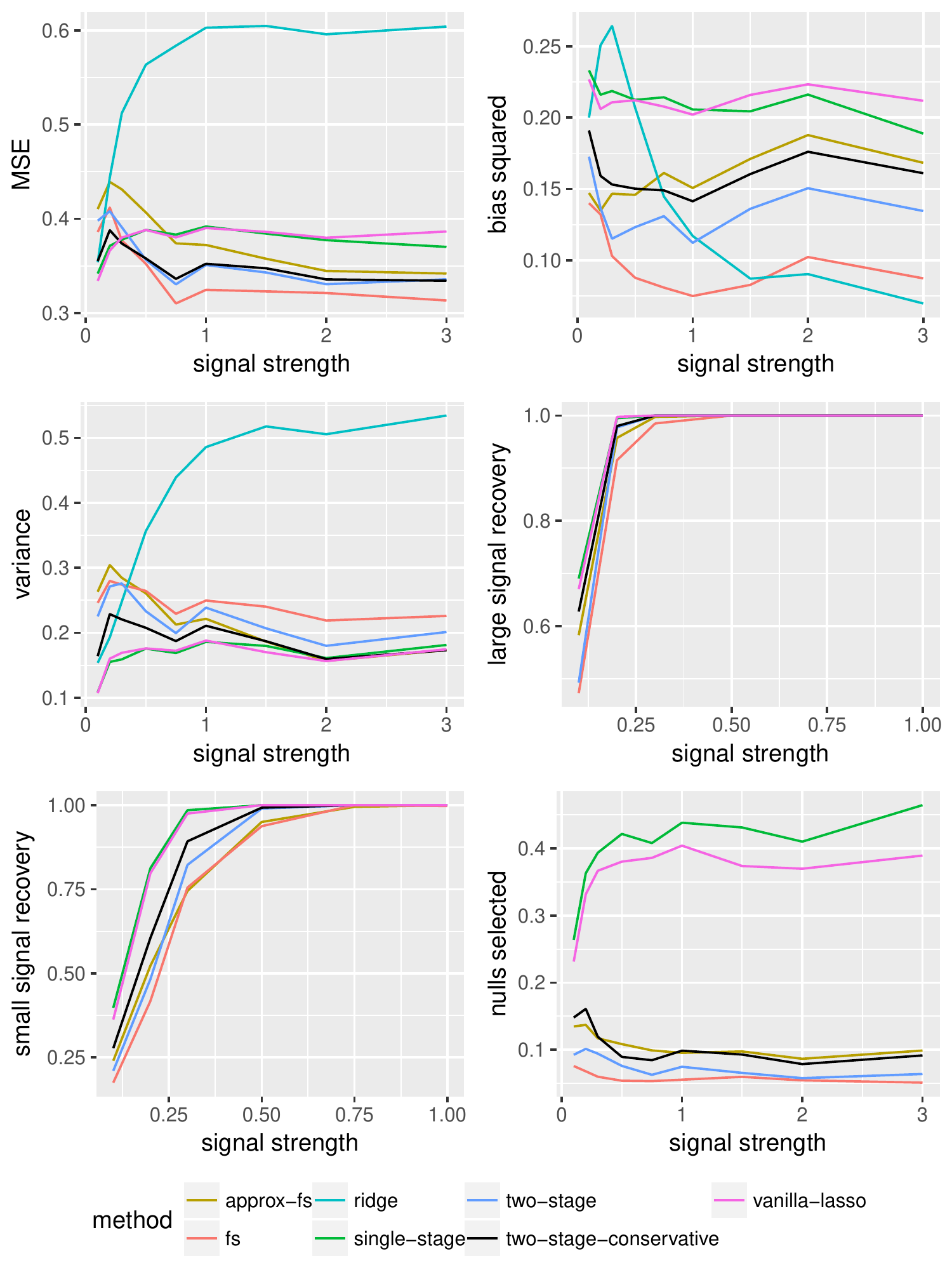}
\caption{{\em Results of experiment 2: performance in the presence of model misspecification.  The ``large signal recovery'' and ``small signal recovery'' graphs report the proportion of times that the true large log-ratio signal and true small log-ratio signal are selected, respectively.  The ``nulls selected'' graph shows the average fraction of null variables that are selected.}}
\label{fig:mispec}
\end{center}
\end{figure}

\subsection{Experiment 3: Computational Efficiency}
Lastly, we compare the computational efficiency of the methods.  The results are presented in figure \ref{fig:runtime-sims}.  We find that approximate forward stepwise selection offers a significant speedup over standard forward stepwise selection.  For a log-ratio model with $p$ individual terms, $k$ steps of approximate forward stepwise will scale as $O(nkp)$ whereas $k$ steps of na\"ive forward stepwise will scale as $O(npk^2)$.

We find that constrained lasso offers a significant speedup over a na\"ive lasso fitting procedure applied to all log-ratio terms.  Recall that by the results in section 2, the solution to these two optimization problems are equivalent.  For a model with $d$ predictors with $p >> n$, lasso has computational complexity approximately $O(nd^2)$.  For a log-ratio model with $p$ raw features, this means that na\"ive lasso fitting will have complexity $O(nd^4)$ whereas constrained lasso fitting will have complexity $O(nd^2)$.

Although the asymptotic scaling of approximate forward stepwise is better than that of the constrained lasso, empirically we find that the constrained lasso procedure is faster even for large problem sizes since the lasso solver within the {\tt glmnet} package is heavily optimized and internally runs routines in FORTRAN.  Nevertheless, the R implementation of the approximate forward stepwise procedure runs in less than 1 second on a notebook computer for a log-ratio model with 500 observations and 500 raw features. Standard forward stepwise regression in R on all log-ratios becomes computationally infeasible for problems of this size.

\begin{figure}
\begin{center}
\includegraphics[scale = .85]{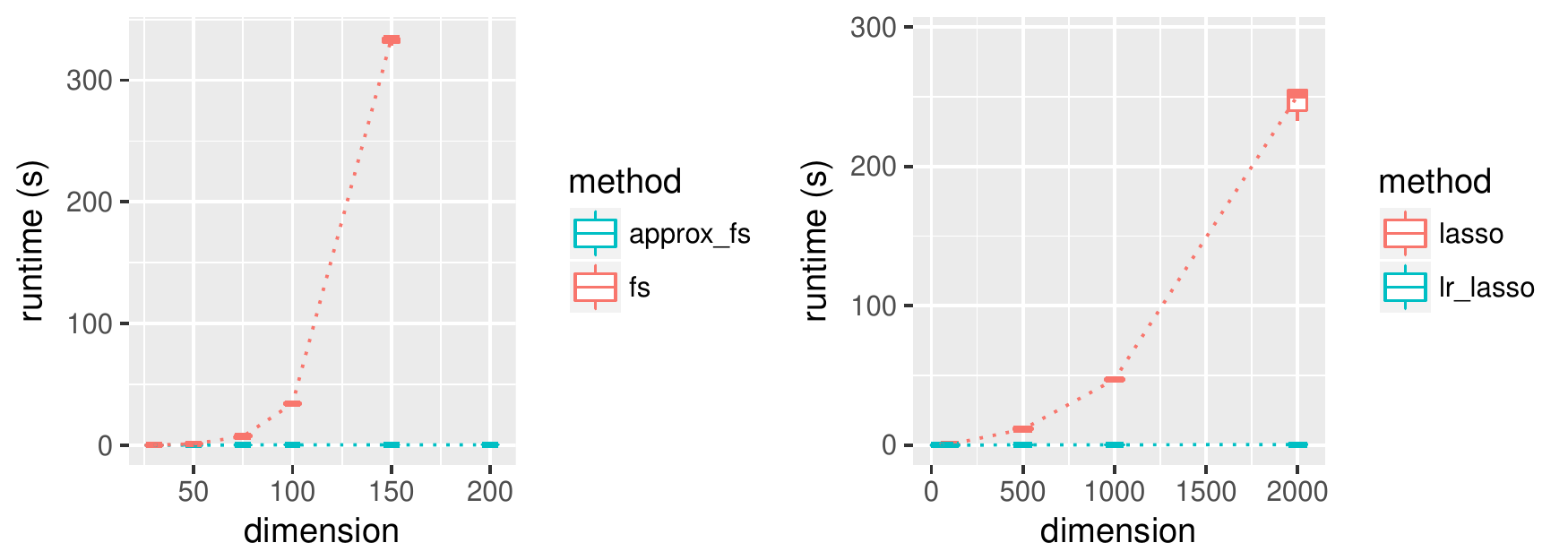}
\end{center}
\caption{{\em Results of experiment 3, a comparison of the runtime of 10 steps of the approximate forward stepwise selection and standard forward stepwise procedure (left) and na\"ive lasso versus constrained lasso fitting (right).  Fitting for forward stepwise selection and and na\"ive lasso are done on the expanded feature set of all log-ratios, which is of size $\binom{\text{dimension}}{2}$. Runtimes are from a Macbook pro with 3.3 GHz Intel Core i7 processor.  Forward stepwise selection was fit using the {\tt leaps}[\cite{leaps}] R package and lasso was fit with the {\tt glmnet}[\cite{glmnet}] R package. We note that {\tt glmnet} is internally running FORTRAN code, which accounts for the large difference in runtime among the methods in the left versus right panels.}}
\label{fig:runtime-sims}
\end{figure}

\section{Real Data Example: Cancer Proteomics} \label{sec:real-data}
\label{sec:real-data}

We now apply the log-ratio lasso methods to a proteomics data set collected and analyzed in \cite{banerjee-et-al}.  The data consist of measurements from 54 patients, each of whom has either a healthy prostate or prostate cancer.  For each individual 5 to 18 locations on a tissue sample are examined, and for each location the intensity of 53 chemical markers is measured using mass spectrometry.  The resulting data set has 618 observations of 53 features.  Researchers are working to use data sets like this to classify tissue samples as healthy or cancerous in minutes, so that they can be used by physicians in the middle of a surgery.  In this data set, observations from the same patient are not independent, and the cross-validation process is done block-wise so that each patient falls entirely in the training fold or test fold.  A validation set of 202 observations from 18 patients is put aside for assessing the accuracy of the final model.  Many of the features are zero, so the value 1 was added to each entry to allow for subsequent logarithmic transformations.

\subsection{Baseline Models}

We perform lasso logistic regression, approximate forward stepwise, and ridge regression on this data set, using a logarithmically transformed feature matrix.  Performance is reported in table \ref{tab:model-results}. Based on previous scientific knowledge, \cite{banerjee-et-al} found that glucose/citrate ratio is a highly predictive feature. Logistic regression on using only this feature and intercept gives 94\% validation set accuracy. We will refer to this model as {\em oracle logistic regression} since it is based on external information.  From this regression model, we see that there is a strong predictive model consisting of only one log-ratio term.  This model is inside the span of the lasso and ridge regression models, but these fitting techniques have very poor predictive performance compared to the oracle model.  Furthermore, lasso logistic regression, does not find this log-ratio in the sense that glucose and citrate do not have the largest positive and smallest negative coefficients.  This is visually represented by the lasso path in figure \ref{fig:lasso-paths}.  A researcher interpreting the fit of the lasso model would not be led to the conclusion that the log-ratio of glucose to citrate is highly important.

\begin{table}
\begin{tabular}{| c c c c |}
\hline
Method & Classification Accuracy & Support Size & Selects Glucose and Citrate? \\
\hline
\hline
Oracle logistic regression & .94 & 2 & Yes \\
\hline
\hline
Ridge regression & .71 & 53 & Yes \\
Approximate FS & .72 & 6 & No \\
Lasso & .73 & 20 & Yes\\
LR lasso (single stage) &.74 & 27  & Yes \\
{\bf LR-lasso (two-stage)} & {\bf .90} & {\bf 4}  & {\bf Yes}\\
\hline
\end{tabular}
\caption{{\em A comparison of predictive accuracy and parsimony of various models.  We note that the log-ratio lasso is coming very close to the performance of the oracle model which
is constructed on the basis of external scientific information.}}
\label{tab:model-results}
\end{table}

\subsection{Log-ratio Lasso}

We next fit the single-stage logistic log-ratio lasso.  The log-ratio lasso results in 74\% validation set accuracy.  In figure \ref{fig:lasso-paths} we see that the lasso path from the single stage log-ratio lasso procedure.  A researcher using this method would clearly see that the ratio of glucose to citrate is highly important, and we will soon see the two stage procedure easily picks this out as the most important feature.

\begin{figure}
\includegraphics[scale = 0.43]{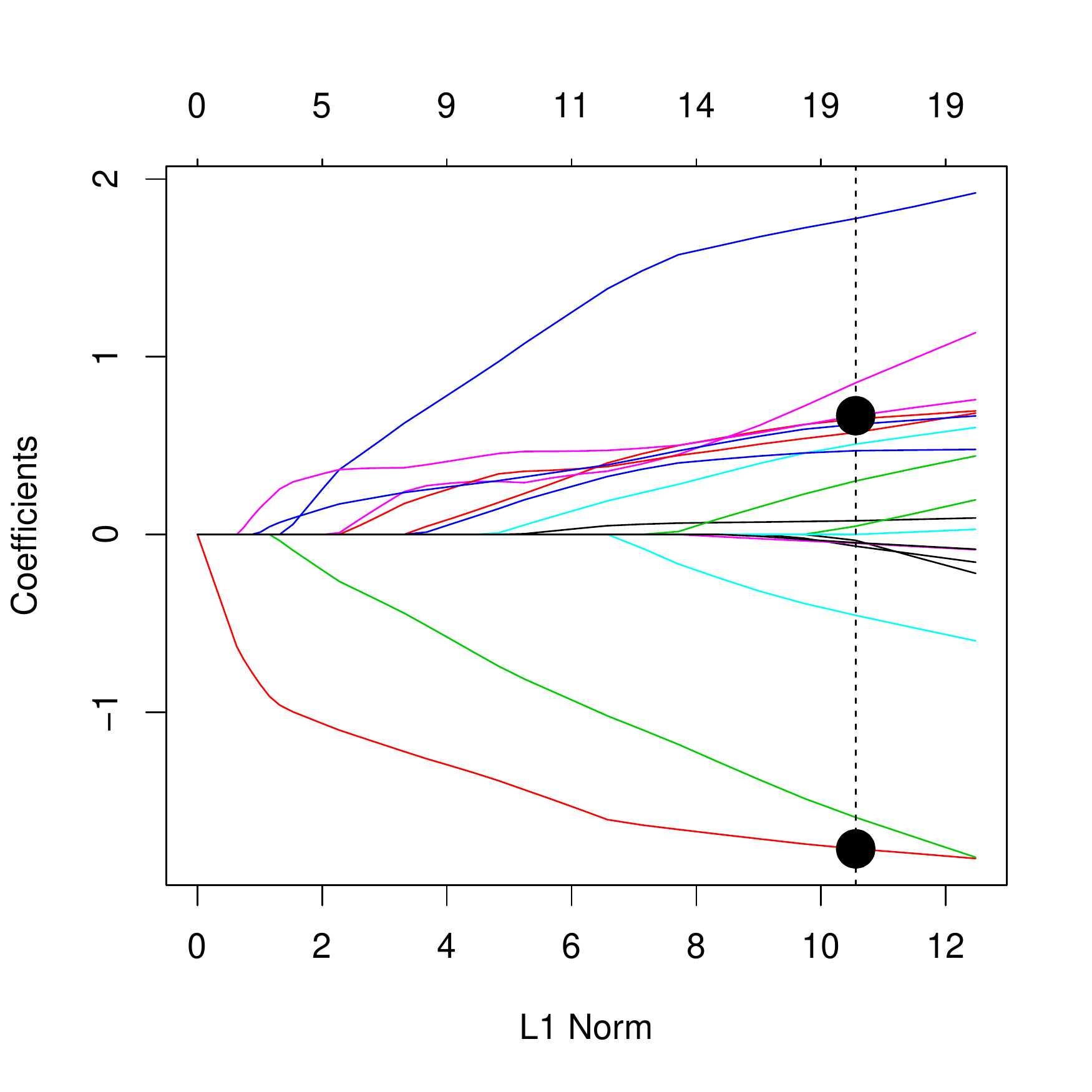}
\includegraphics[scale = 0.43]{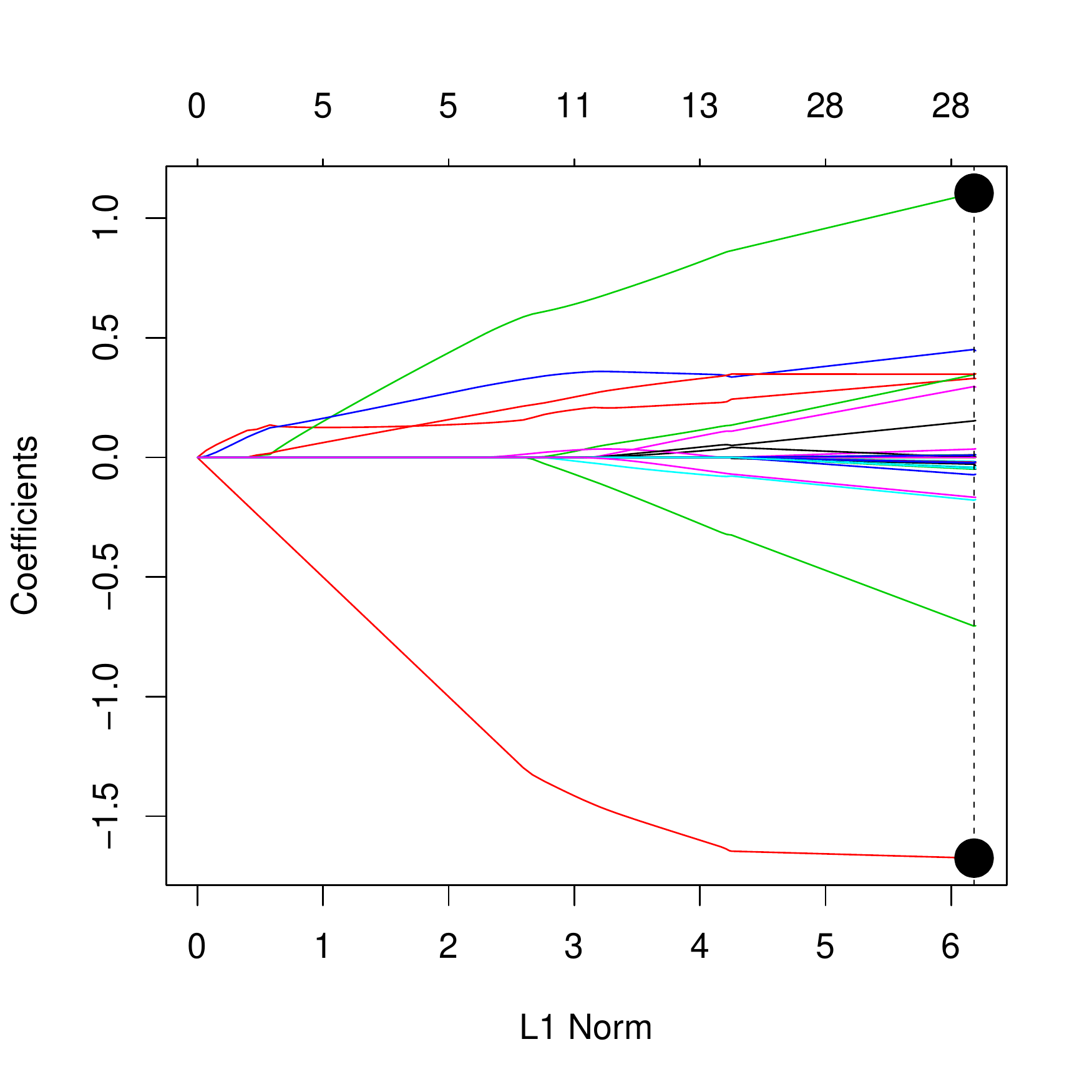}
\caption{{\em A comparison of the selection paths from lasso logistic regression (left) and the single stage log-ratio lasso (right). The top horizontal labels indicate how many variables are in the model at each point along the path. Dashed vertical lines indicate the tuning parameter selected by cross-validation. The coefficients of glucose and citrate for the optimal vlaue of the tuning parameter are marked with large circles.  Notice that citrate is not easily picked out on the left plot, but it is easily picked out on the right.}}
\label{fig:lasso-paths}
\end{figure}

Lastly, we fit the two-stage log-ratio lasso.  This model selects two log-ratio terms, presented in table \ref{tab:lrl-model}.  The first term selected is the oracle model term previously reported in the literature: the logarithm of glucose divided by citrate.  In addition to being highly parsimonious, the predictive performance of the model is close to that of the oracle model, and much better than that of the baseline methods.  Figure \ref{fig:box-roc} shows that the predictions from the log-ratio lasso on the validation set.  The two-stage log-ratio lasso procedure nicely separate the two classes, and the ROC curve shows that the predictive performance is much better for any choice of the false positive rate.  The model rivals that of the oracle model in both model parsimony and predictive accuracy. It is highly appealing the log-ratio lasso is able to systematically find such a model fit.

\begin{table}
\begin{center}
\begin{tabular}{| c c |}
\hline
Feature & Coefficient \\
\hline
$\log({X_{glucose} /  X_{citrate/isocitrate})}$ & $.24$ \\
$\log({X_{[glucose/fructose]Cl-} /X_{oleic \ acid \ dimer})}$ & $-.09$ \\
\hline
\end{tabular}
\end{center}
\caption{{\em Final model selected by two-stage log-ratio lasso}.}
\label{tab:lrl-model}
\end{table}

\begin{figure}
\begin{center}
\includegraphics[scale = .8]{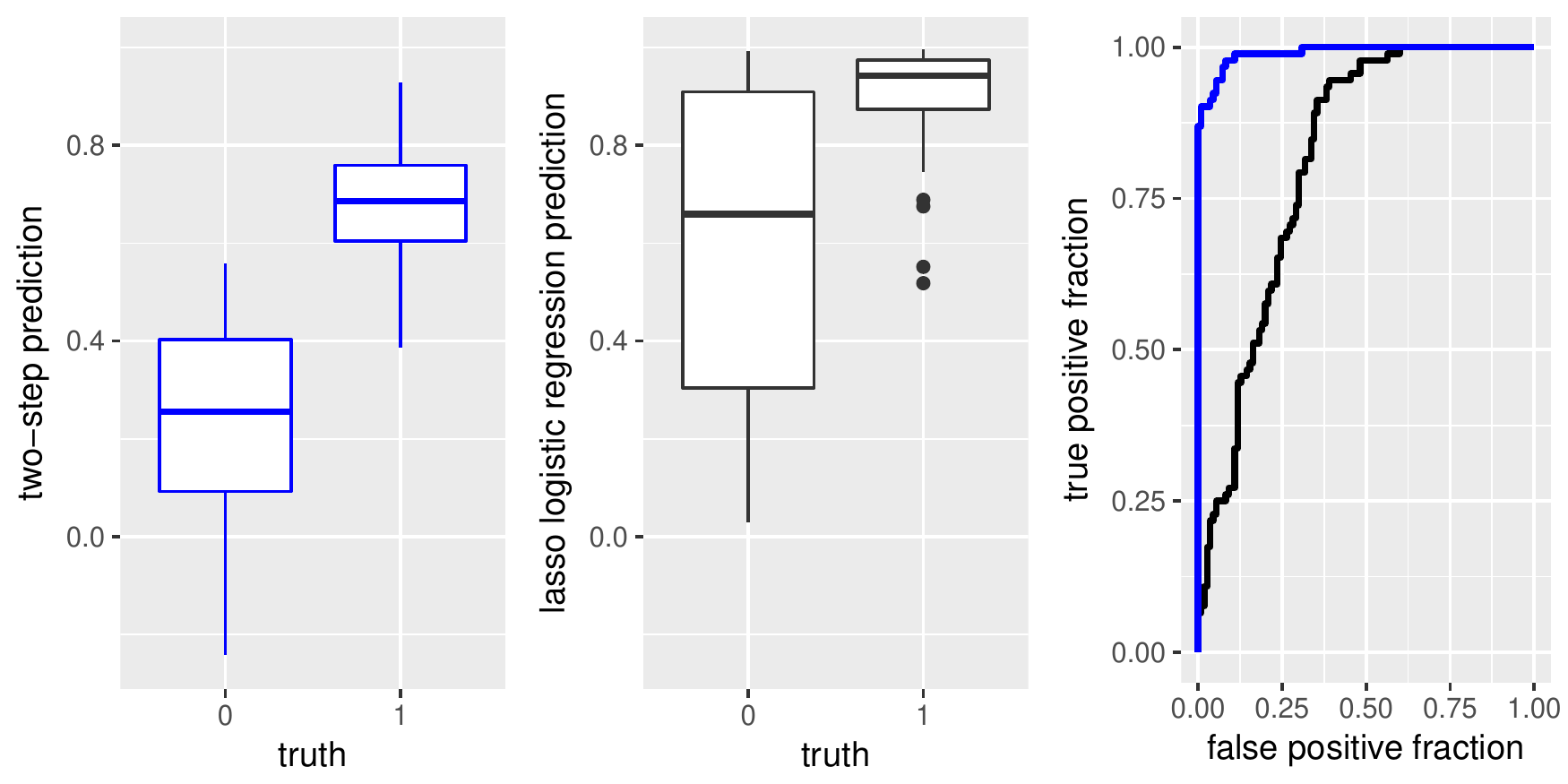}
\caption{{\em A comparison of the predictions on a validation set generated by lasso logistic regression (black) and two-stage log-ratio lasso (blue) using box plots and ROC curves.  The AUC is .81 for lasso logistic regression and .99 for two-stage log-ratio lasso.}}
\label{fig:box-roc}
\end{center}
\end{figure}


\section{Discussion} \label{sec:discussion}
We have formulated a new and useful notion of sparsity for compositional data based on the log-ratio model (\ref{eq:logratio}). We have introduced a novel, principled variable selection procedure for this models. We prove the equivalence of our method to a constrained lasso problem with a small number of variables. To the best of our knowledge, there are no existing specialized variable selection procedures for this model, and the na\"ive application of standard techniques would require the creation of $\binom{p}{2}$ additional features, which quickly makes the runtime and storage requirements prohibitively large. In contrast, the method introduced in this work only requires solving a modified lasso optimization with $p$ features, which enjoys favorable runtime and storage even for very large $p$. We extend this method with a second pruning step which leads to highly sparse models and greatly improves the performance in simulation experiments and on a real data set. On real data, the method recovers a very sparse model containing features of known relevance and with high predictive accuracy. The method appears to be very well-suited for imaging data in biological and medical domains, where the relative intensity of the raw features are the scientifically meaningful quantities, and where researchers are often jointly interested in predictive accuracy and model parsimony.

\subsubsection*{Future directions}
We conclude by pointing to a few directions for further work.
 \begin{itemize}
 \item {\em Extensions beyond the linear model.}  In section \ref{sec:lrl-estimator} we mentioned that the low-dimensional characterization of the log-ratio lasso holds for other models such as GLMs and the Cox proportional hazards model. Since binary outcome data and censored time data are common in medical imaging applications, it would be of great interest to extend the methodology introduced in this paper to these cases.
 \item {\em Sophisticated second stage pruning.}  In subsection \ref{subsec:two-stage} we saw that $L_1$-penalization is not sufficient to enforce adequate sparsity in the log-ratio model (\ref{eq:logratio}).  We added a second pruning stage to pair the terms into log-ratios.  In this work, we concentrated on forward stepwise selection, but there are other useful sparse regression procedures, typically cast as a non-convex optimization problem.  It would be of interest to see if using these more sophisticated methods in the pruning stage leads to better performance.
 \item {\em Selective inference.} A flurry of recent work in the statistical literature has introduced methods for post-selective tests and confidence intervals. Theoretical tools developed to provide post-selective confidence intervals for the usual lasso estimator can likely be used to develop post-selective tests and confidence intervals for both the single-stage and two-stage log-ratio lasso estimator.
 \end{itemize}


\subsection*{Acknowledgments}
Stephen Bates was supported by NIH grant T32 GM096982. Robert Tibshirani was supported by NIH grant 5R01 EB001988-16 and NSF grant 19 DMS1208164.  

\newpage
\bibliographystyle{agsm}
\bibliography{references}

\newpage \appendix
\section{Proof of Theorem \ref{thm:main}} \label{ap:proofs}
The following propositions together give proof of the theorem.  Let $b: \R^{\binom{p}{2}} \to \R^p$ be the map that takes a $\theta$ from a log-ratio lasso feature space to the corresponding $\beta$ in the standard feature space:
\begin{equation*}
b(\theta)_k = \sum_{j=1}^{k-1} -\theta_{j,k} + \sum_{j = k+1}^p \theta_{k,j}.
\end{equation*}

\begin{prop}
For $\beta = b(\theta)$ we have $\sum_{k=1}^p \beta_k = 0$.
\end{prop}

\begin{proof}
\begin{align*}
\sum_{k=1}^p \beta_k &= \sum_{k=1}^p [\sum_{j = 1}^{k-1} -\theta_{j,k} + \sum_{j = k+ 1}^p \theta_{k,j}] \\
&= \sum_{1 \le j < k \le p} -\theta_{j,k} + \sum_{1 \le j < k \le p} \theta_{j,k} \\
&= 0 
\end{align*}
\end{proof}

\begin{prop}
The model corresponding to $\beta = b(\theta)$ and the model corresponding to $\theta$ have the same sum of squared residuals.
\end{prop}

\begin{proof}
\begin{align*}
\sum_{k=1}^p \beta_k \log(x_{i,k}) &= \sum_{k=1}^p[\sum_{j=1}^{k-1} -\theta_{j,k} + \sum_{j = k + 1}^p \theta_{j,k}] \log(x_{i,k})\\
	&= \sum_{1 \le j < k \le p} - \theta_{j,k} \log(x_{i,k}) + \theta_{j,k}\log(x_{i,j})\\
	&= \sum_{1 \le j < k \le p}\theta_{j,k}\log{\frac{x_{i,j}}{x_{i,k}}}
\end{align*}
Thus the two models have the same fitted value for each observation $l = 1,...,n$, and hence the same sum of squared residuals.
\end{proof}

\begin{prop}
For any $\beta$ such that $\sum_{k=1}^p \beta_k =0$, there exists a $\theta$ such that $\beta = b(\theta)$ with the property that $\norm{\beta}_1 = 2 \norm{\theta}_1$.
\end{prop}

\begin{proof}
Without loss of generality, suppose $\beta_1, ..., \beta_{p^+} \geq 0$ and $\beta_{p^+ + 1}, ..., \beta_p < 0$. Let $\theta_{i,j} = 0$ if $i,j \leq p^+$ or if $i,j > p^+$. For $1 \leq i \leq p^+ < j \leq p$, let $\theta_{i,j} = \frac{2|\beta_i||\beta_j|}{\norm{\beta}_1}$.

Now $\beta_k = \sum_{i=1}^{k-1} -\theta_{i,k} + \sum_{i=k+1}^p\theta_{k,i}$ so for $k \leq p^+$ we have:
\begin{align*}
b(\theta)_k &= \sum_{i = 1}^{k-1}-\theta_{i,k} + \sum_{i=k+1}^p\theta_{k,i} \\
&= \sum_{i = p^+ +1}^p \theta_{k,i} \\
&= \sum_{i = p^+ + 1}^p \frac{2|\beta_i||\beta_k|}{\norm{\beta}_1} \\
&= 2 \beta_k \sum_{i = p^+ + 1}^p \frac{|\beta_i|}{\norm{\beta}_1} \\
&= \beta_k.
\end{align*}
The last equality follows from the fact that $\norm{\theta}_1 = \sum_{j=1}^p |\beta_j| = \sum_{j=1}^{p+} \beta_j - \sum_{j=p^+ + 1}^p \beta_j$ and $\sum_{j=1}^p \beta_j = 0$.  The analogous computation holds for $k > p^+$, so $\beta = b(\theta)$.

Now notice:
\begin{align*}
\norm{\beta}_1 &= \sum_{k=1}^p|\beta_k| \\
&= \sum_{k = 1}^p |\sum_{i = 1}^{k-1} -\theta_{i,k} + \sum_{i = k+1}^p \theta_{k,i}| \\
&= \sum_{k = 1}^{p^+} |\sum_{i = p^+ + 1}^p \theta_{k,i}| + \sum_{k = p^+ + 1}^p|\sum_{i = 1}^{p^+} -\theta_{i,k}| \\
&= \sum_{1 \le k < i \le p}|\theta_{k,i}| + \sum_{1 \le i < k \le p}|-\theta_{i,k}| \\
&= 2 \norm{\theta}_1. 
\end{align*}
\end{proof}

\begin{remark}
We can see from this last proof that the log-ratio lasso fit is not identifiable: many different values of $\theta$ correspond to both the same fit and the same 1-norm penalty.
\end{remark}

\begin{prop}
Suppose $\theta$ is a solution to the log-ratio lasso optimization problem, and let $\beta = b(\theta)$.  Then $\norm{\beta}_1 = 2 \norm{\theta}_1$.
\end{prop}

\begin{proof}
Suppose not.  Then $|\beta_k| \neq \sum_{i < k}|-\theta_{i,k}| + \sum_{k > i}|\theta_{k,i}|$ for some $k$.  We will now show that the 1-norm of $\theta$ can be reduced without changing the fitted values, which is a contradiction. Suppose without loss of generality that there exist $i < j < k$ such that $\theta_{i,k} < 0 < \theta_{j,k}$ with $|\theta_{i,k}| > |\theta_{j,k}|$. Then consider a new fit $\tilde{\theta}$ with $\tilde{\theta} = \theta$ except for the following:
\begin{align*}
\tilde{\theta}_{i,k} &= \theta_{i,k} + \theta_{j,k}\\ 
\tilde{\theta}_{j,k} &= 0 \\
\tilde{\theta}_{i,j} &= \theta_{i,j} + \theta_{j,k}.
\end{align*}
$\tilde{\theta}$ results in the same fit as $\theta$ but has a 1-norm reduced by $\theta_{j,k}$.
\end{proof}

Combining these four propositions proves the main theorem.

\section{Post-Selective Inference Technical Details}
We will state the relevant technical results from \cite{lee-et-al} and then customize them for our setting.  Let $\hat{M}$ be the support set and signs selected by the lasso.  \cite{post-selection} and \cite{lee-et-al} establish that the event $\{ \hat{M} = M\}$ can be expressed as a polyhedron:
\begin{equation} \label{lasso-polyhedron}
\{\hat{M} = M\} = \{A(M,s)y \le b(M, s)\}.
\end{equation}
The matrices $A(M,s)$ and vectors $b(M,s)$ are given by the following:
\begin{align*}
A(M,s) &:= 
	\begin{bmatrix}
		\frac{1}{\lambda} X_{-M}^\top (I - P_M) \\
		\frac{-1}{\lambda} X_{-M}^\top (I - P_M) \\
		-\diag(s)(X_M^\top X_M)^{-1} X_M^\top 
	\end{bmatrix} \\
b(M,s) &:= 
	\begin{bmatrix}
		1 - X_{-M}^\top (X_M^\top X_M)^{-1} X_M^\top s \\
		1 + X_{-M}^\top (X_M^\top X_M)^{-1} X_M^\top s \\
		-\lambda \diag(s) (X_M^\top X_M)^{-1} s
	\end{bmatrix}
\end{align*}
where $P_M$ denotes the orthogonal projection onto the column span of $X_M$. Using this result, \cite{lee-et-al} compute the conditional distribution of $\eta_M^\top y$ given $\{\hat{M} = M\}$ for any vector $\eta_M$.  That work explicitly treats the case where $\eta_M$ is chosen to test hypotheses about the partial regression coefficients $\beta^{(M)}_j$, which is often interest.  For our setting, we instead use these results to test whether a log-ratio model is consistent with the observed lasso fit, which we formulated as a formal hypothesis in (\ref{eq:selective-hypothesis}).  Taking $\eta_M = 1_M^\top (X_M^\top X_M)^{-1} X_M^\top$, we have that
\begin{align*}
\eta_M^\top \E[y] &= 1_M^\top (X_M^\top X_M)^{-1} X_M^\top \E[y] \\
	&= 1_M^\top \E [(X_M^\top X_M)^{-1} X_M^\top y] \\
	&= 1_M^\top \beta_M. 
\end{align*}
Thus, this choice of $\eta_M$ corresponds to testing the hypothesis in (\ref{eq:selective-hypothesis}).  From here, an application of the \cite{lee-et-al} machinery yields a pivotal quantity for $\eta_M^\top y$ after conditioning on $\{\hat{M} = M\}$, which we encapsulate in the following proposition.
\begin{proposition}[Post-selective test of the log-ratio model, detailed version]
Let $F_{\mu, \sigma^2}^{[a,b]}$ be the CDF of a $N(\mu, \sigma^2)$ random variable truncated to the set $[a,b]$. Let $z := (I - P_{\eta_M})y$ be the residual of the projection of $y$ onto $\eta_M$, which is independent of $\eta_M^\top y$, and let $c := \frac{\eta_M}{\norm{\eta_M}^2}$.  Define:
\begin{align*}
V^+_{M,s}(z) &:= \max_{j: (A(M,s)c)_j < 0} \frac{b(M,s)_j - (A(M,s)z)_j}{(A(M,s)c)_j}\\
V^-_{M,s}(z) &:= \min_{j: (A(M,s)c)_j > 0} \frac{b(M,s)_j - (A(M,s)z)_j}{(A(M,s)c)_j}.
\end{align*}
Then the following holds:
\begin{equation*}
F_{1^\top \beta^{(M)}, \norm{\eta_M}^2}^{[V_{M,s}^-(z) ,V_{M,s}^+(z)]}(\eta_M^\top y) | \{\hat{M} = M\} \sim \text{Unif}(0,1).
\end{equation*}
\end{proposition}
\begin{proof}
This is an application of Theorem 5.3 of \cite{lee-et-al} with $\eta_M = 1_M^\top (X_M^\top X_M)^{-1} X_M$ using the characterization of the lasso selection event in (\ref{lasso-polyhedron}).
\end{proof}

\section{Solving the Constrained Lasso Optimization Problem} \label{ap:optimization}
The constrained lasso optimization problem given in equation \ref{eq:constr-lasso} is a convex optimization problem in $p$ variables.  It can be cast as an optimization problem with a quadratic objective function in $2p$ variables with only linear inequality constraints and a single linear equality constraint: 
\begin{align*}
\minimize_{\beta^+_1, \beta^-_1 ,..., \beta^+_p, \beta^-_p} \ \  & \frac{1}{2} \sum_{i=1}^n[\sum_{j=1}^p (y_i - \beta^+_j \log(x_{i,j}) + \beta^-_j \log(x_i,j))^2] + \lambda(\sum_{j=1}^p \beta^+_j + \beta^-_j) \\
\text{subject to \ \ } & \beta^+_j \ge 0 \text{ \ for \ } j=1,...,p \\
	& \beta^-_j \ge 0 \text{ \ for \ } j=1,...,p \\
	&\sum_{j = 1}^p \beta^+_j - \beta^-_j = 0.
\end{align*}
Such an optimization problem can be efficiently solved with standard optimization libraries such as the popular open-sourced {\tt CVX}[\cite{cvx}] for MATLAB or {\tt CVXPY}[\cite{cvxpy}] for python.  

The constrained lasso optimization problem can also be solved efficiently using lasso solvers such as ${\tt glmnet}$ which allow for weighted observations. One simply augments the data with an additional data point with all features equal to 
$1$, and response value zero. By assigning this value a large weight, the resulting solution $\beta$ will have $\sum_{j=1}^p \beta_j \approx 0$.  The value of $\sum_{j=1}^p \beta_j$ can be made arbitrarily small with large values of the weight.  Similarly, for the logistic regression analog of the constrained lasso, one simply augments the feature matrix with two entries of large equal weight.  One entry is assigned value 1 to all features and value 1 to the response.  The other entry is assigned value 1 to all features and value 0 to the response.  Because dedicated lasso solvers use specialized tricks to improve performance, this approach will typically be much faster than using a general-purpose convex optimization solver. \cite{constrained-lasso} give a detailed analysis of a coordinate descent algorithm for the constrained lasso that is similar to this proposal in the special case where the lasso solver is using coordinate descent [\cite{glmnet}].

\end{document}